\documentclass[12pt]{article}
\usepackage{amsmath, amsthm}
\usepackage{graphicx}
\usepackage{mathtools}
\usepackage{caption}
\usepackage{subcaption}
\usepackage{mwe}
\newtheoremstyle{theorem}%name
  {10pt}		  % space above
  {10pt}  % space below
  {\sl}  % bofy font
  {\parindent}     % ident - empty=no indent,  \parindent= paragraph indent
  {\bf}  % thm head font
  {. }    % punctuation after thm head
  { }    % space after thm head: `` ``=normal \newline=linebreak
  {}     % thm head specification
\theoremstyle{theorem}
\newtheorem{theorem}{Theorem}

\newtheoremstyle{defi}%name
  {10pt}		  % space above
  {10pt}  % space below
  {\rm}  % bofy font
  {\parindent}     % ident - empty=no indent,  \parindent= paragraph indent
  {\bf}  % thm head font
  {. }    % punctuation after thm head
  { }    % space after thm head: `` ``=normal \newline=linebreak
  {}     % thm head specification
\theoremstyle{defi}

\begin{document}

\title{Language properties and Grammar of Parallel and Series Parallel Languages}

\author{Mohana.N$^1$, Kalyani Desikan$^2$ and V.Rajkumar Dare$^3$\\
$^{1}$ Division of Mathematics, School of Advanced Sciences,\\
VIT University, Chennai, India\\
mohana.n@vit.ac.in\\[2pt]
$^2$Division of Mathematics, School of Advanced Sciences,\\
VIT University, Chennai, India\\
kalyanidesikan@vit.ac.in\\[2pt]
$^3$Department of Mathematics, Madras Christian College,\\
Chennai, India\\
rajkumardare@yahoo.com}
\date{}
\maketitle

\begin{abstract}
In this paper we have defined the language theoretical properties of Parallel languages and series parallel languages. Parallel languages and Series parallel languages play vital roles in parallel processing and many applications in computer programming. We have defined regular expressions and context free grammar for parallel and series parallel languages based on sequential languages [2]. We have also discussed the recognizability of parallel and series parallel languages using regular expression and regular grammar.
\end{abstract}

{\bf Keywords: Formal language theory, Series parallel languages, Branching automaton, context free language, regular grammar}
%\linenumbers
\section{Introduction}
A language is a medium of communication. In communicating a problem to a machine, the design of a proper language of computation is important and this is the fundamental objective of computability. From the perspective of theory of computation, any type of problem can be expressed in terms of language recognition.\\ Fundamentally, a computer is a symbol manipulator. It takes sequences of symbols as input and processes them as per the program specifications. Finite sequence of symbols over an alphabet is called a string. In other words, alphabets are sequentially arranged. If the alphabets are arranged in parallel, then we call it as parallel words or strings. Series parallel words are arranged both sequentially and in parallel. In this paper we have given the basic notations and definitions of parallel and series parallel words. In section 3, we have discussed about operations such as concatenation, parallel operation, Kleene closure and String Reversal on parallel and series parallel languages. We have also defined context free grammar on parallel and series parallel languages in section 4. Regular expressions and regular languages have been defined in section 5 and 6. In section 7, we have given parallel regular grammar and series parallel regular grammar and their properties.
\section {Preliminaries}
Let $\Sigma$ be an alphabet. Let $\Sigma^{*}$ denote the set of all finite sequential terms and $\Sigma^{\oplus}$ the set of all finite parallel terms over $\Sigma$. In general, $\Sigma_{p}^{+}=\{x\|y\colon x,y\in\Sigma^{+}\}$, $\Sigma_{s}^{\oplus}=\{x.y\colon x,y\in\Sigma^{\oplus}\}$.
$SP$($\Sigma$) is the set of all finite series parallel words over $\Sigma$. In other words $SP$($\Sigma$)=$\Sigma^{+}$$\cup$$\Sigma^{\oplus}$$\cup$$\Sigma_{p}^{+}$$\cup$$\Sigma_{s}^{\oplus}$.\\ Language is a set of words or strings. If $L\subseteq \Sigma^{*}$, $L\subseteq \Sigma^{\oplus}$ and $L\subseteq SP^(\Sigma)$ then L is said to be a sequential, parallel and series parallel languages over $\Sigma$, respectively.\\ Length of a word from $\Sigma^{*}$ can be defined as the number of alphabets in the word and the length of a word from $\Sigma^{\oplus}$ is always one. Similarly, depth of a word from $\Sigma^{\oplus}$ is the number of alphabets in the word and the depth of the word from $\Sigma^{*}$ is one. In general, length and depth of a word in a series parallel language can be defined as follows:\\ $lg(x.y)=lg(x)+lg(y)$, $lg(x\|y)=max(lg(x),lg(y))$\\$dp(x.y)=max(dp(x),dp(y))$, $dp(x\|y)=dp(x)+dp(y)$\\ where $lg$ and $dp$ represent the length and depth of a word.
\section{Operations on languages}
We have some basic operations such as concatenation, parallel operation, Kleene Closure and String Reversal on languages over $\Sigma$. Here we discuss these operations for parallel languages and series parallel languages. The above operations on sequential languages have already been discussed in [2].
\subsection{Parallel languages}
Let $L1,L2\subseteq \Sigma^{\oplus}$.\\$\mathbf{Concatenation}$ of two parallel languages $L1$ and $L2$ is defined as $L=L1\cdot L2$. That is $L=\{x\cdot y|x\in L1, y\in L2\}$.\\ $\mathbf{Parallel}$ $\mathbf{Operation}$ on two parallel languages $L1$ and $L2$ is given by $L=L1\| L2$. That is $L=\{x\| y|x\in L1, y\in L2\}$.\\
Let $L$ be a set of strings from $\Sigma^{\oplus}$. $\mathbf{Kleene}$ $\mathbf{Closure}$ of $L$ is defined as the set of strings formed by performing the parallel operation on strings from $L$ with repetitions. More generally, $L^{\oplus}$ is an infinite union $\cup L_{n}$, for $n\geq 0$ where $L_{n}$ represents $n$ number of repetitions applied parallely on  strings of $L$.\\ For instance, if $L=\{a,a\|b\}$ the $L_{0}=\epsilon$, $L_{1}=\{a,a\|b\}$, $L_{2}=\{a\|a,a\|a\|b,a\|b\|a,\\a\|b\|a\|b\}$ and so on.
%sequential languages as $S^{*}=\cup S^{i}$, for $i\geq0$, where $S^{i}$ denotes a set of strings of length $i$. For parallel languages $S^{\oplus}=\cup S_{i}$,for $i\geq0$, where $S_{i}$ denotes a set of strings of depth $i$.\\
\\Let $x$ be a string in $\Sigma^{\oplus}$. Then $\mathbf{String}$ $\mathbf{reversal}$ (read backwards) of $x$ is denoted as $x^{R}$ which satisfies the conditions $x^{R}=x$ with $dp(x^{R})=dp(x)$ and $lg(x^{R})=lg(x)$.
\subsection{Series parallel languages}
Let $L1,L2\subseteq SP(\Sigma)$.\\$\mathbf{Concatenation}$ of two series parallel languages $L1$ and $L2$ is defined as $L=L1\cdot L2$. That is $L=\{x\cdot y|x\in L1, y\in L2\}$.\\ $\mathbf{Parallel}$ $\mathbf{Operation}$ on two series parallel languages $L1$ and $L2$ is given by $L=L1\| L2$. \\That is $L=\{x\| y|x\in L1, y\in L2\}$.\\
Let $L$ be a set of strings from $SP(\Sigma)$. $\mathbf{Kleene}$ $\mathbf{Closure}$ of $L$ is defined as $L^{\otimes}=L^{*}\cup L^{\oplus}$, where $L^{*}=\cup L^{n}$ for $n\geq0$, $L^{n}$ represents $n$ number of repetitions and concatenation of strings of $L$ and $L^{\otimes}$ indicates the set of series parallel strings.\\ For example, let $L=\{ab,a\|b\}$ then $L^{0}=L_{0}=\epsilon$, $L^{1}=L_{1}=L$,\\ $L^{2}=\{abab,ab(a\|b),(a\|b)ab,(a\|b)(a\|b)\}$, $L_{2}=\{(ab)\|(ab),(ab)\|(a\|b),(a\|b)\|(ab),\\(a\|b)\|(a\|b)\}$ and so on.
\\Let $x$ be a string in $SP(\Sigma)$. Then $\mathbf{String}$ $\mathbf{reversal}$ of $x$ is defined as $x^{R}$ which satisfies the conditions $dp(x^{R})=dp(x)$ and $lg(x^{R})=lg(x)$.
\section{Context Free Grammar}
We use grammar to generate words of the language and it is represented by the set variables and terminals. The rules relating the variables are called productions.\\ In [2], we have a context free grammar, $G$ for sequential language is defined as $G=(V,T,P,S)$, where $V$ is a finite set of variables or non-terminals, $T$ is a finite set of terminals, $P$ is a finite set of production rules and $S$ is a start symbol, $S\in V$.\\ Each production is of the form $A\rightarrow \alpha$, where $A\in V$, $\alpha \in (V\cup T)^{*}$.\\A language generated by a context free grammar $G$ is called context free language. That is, $L(G)=\{w|w\in T^{*}, S\Rightarrow w\}$
\subsection{Context free Parallel languages}
Let $G=(V,T,P,S)$ a context free grammar. If $P$ has productions of the form $A\rightarrow \alpha$, $A\in V$, $\alpha \in (V\cup T)^{\oplus}$ then $G$ is said to be a context free parallel grammar.\\
A language generated by a context free parallel grammar $G$ with the production rules as described above is called context free parallel language. That is, $L(G)=\{w|w\in T^{\oplus}, S\Rightarrow w\}$\\
\emph{Example:1} Consider the grammar $G=\{\{S\},\{a,b\},P,S\}$ with production rules\\
\begin{center}
$S\rightarrow a\|b\|S$\\ $S\rightarrow a\|b|\epsilon$
\end{center}
This generates the language $L(G)=\{(a\|b)^{n^{\oplus}},n>0\}$ where $n^{\oplus}$ represents parallel iterations.
\subsection{Context free series parallel languages}
Let $G=(V,T,P,S)$ a context free grammar. If $P$ has productions $A\rightarrow \alpha$, $A\in V$, $\alpha \in SP(V\cup T)$ then G is called as context free series parallel grammar.\\ A language generated by a context free series parallel grammar $G$ with the production rules as described above is called context free series parallel language.
\emph{Example:2} Consider the grammar $G=\{\{A,B\},\{a,b\},P,S\}$ with production rules\\
\begin{center}
$S\rightarrow aA\|bB$,\\ $A\rightarrow Aa|\epsilon$\\ $B\rightarrow bB|\epsilon$
\end{center}
The above grammar generates the language $L(G)=\{a^{m}\|b^{n}|m,n\geq 1\}$
\section {Regular Expressions}
Regular expression is an another way of defining a language. Regular expressions have been defined by algebraic laws of arithmetic in [2] for sequential words over $\Sigma$. Now we define regular expressions for parallel and series parallel words as follows:
\subsection{Parallel Regular Expressions}
Let $\Sigma$ be an alphabet.
\begin{itemize}
\item $\phi$ is a regular expression
\item $a\in \Sigma$ is a regular expression
\item If $R$ is a regular expression then $R^{\oplus}$ is a regular expression
\item If $R1$ and $R2$ are regular expressions then $R1.R2$,$R1\cup R2$ and $R1\|R2$ are also regular expressions.
\end{itemize}
\subsection{Series Parallel Regular Expressions}
Regular expressions for series parallel strings can be defined using the same conditions as that of parallel strings, but instead of $R^{\oplus}$, we have $R^{\otimes}$ to describe Kleene closure on $SP(\Sigma)$.
\section{Regular Languages}
We have defined parallel regular languages and series parallel regular languages based on the definition of sequential regular languages in [1] as follows:
\begin{itemize}
\item $\phi$, $\{a\}$, $\bigcup L_{n}$, $L_{1}\cdot L_{2}\ldots L_{n}$, $L_{1}\|L_{2}\|...\|L_{n}$ and $L^{\oplus}$ are regular languages, where $L\subseteq \Sigma^{\oplus}$
\item $\phi$, $\{a\}$, $\bigcup L_{n}$, $L_{1}\cdot L_{2}\ldots L_{n}$, $L_{1}\|L_{2}\|...\|L_{n}$, $L^{*}$ and $L^{\oplus}$ are regular languages, where $L\subseteq SP(\Sigma)$
\end{itemize}
We can also define a regular language as a language recognized by an automaton and a language generated by a regular grammar as proved in [2] for sequential languages.\\
Recognizability of parallel and series parallel languages on Branching automaton has been discussed in [1] and [6]. Now we define regular grammar for parallel and series parallel languages and prove their regularity.
\section{Regular Grammar}
A grammar $G$ is said to be regular if $G$ is either right-linear or left-linear.\\ Consider $G=(V,T,P,S)$ the production rule $A\rightarrow xB$ or $B\rightarrow x$ is called right-linear and the production rule $A\rightarrow Bx$ or $B\rightarrow x$ is called left-linear, where $A,B\in V$ , $x\in T^{*}$.
\subsection{Parallel Regular Grammar}
A grammar $G$ on parallel languages is said to be regular [2] if $G$ is linear.\\ Consider $G=(V,T,P,S)$, the production rule $A\rightarrow x\|B$ or  $A\rightarrow B\|x$ or $B\rightarrow x$ is called linear, where $A,B\in V$ , $x\in T^{\oplus}$.\\
\begin{theorem}
Let $G=\{V,T,P,S\}$ be a linear grammar. Then $L(G)$ is a parallel regular language.
\end{theorem}
\begin{proof}
By the definition of linear grammar G the production rules are $A\rightarrow x\|B | B\|x$, $B\rightarrow x$ where $A,B\in V$ , $x\in T^{\oplus}$.\\ Now the language generated by the linear grammar is $L(G)=\{w|w\in T^{\oplus}\}$. Then by the definition of parallel regular languages, $L(G)$ is a parallel regular language.
\end{proof}
$\emph{Example:}$ Let $G=(V,T,P,S)$ where $V=\{A,B,S\}$, $T=\{b\}$ and \\$P=\{S\rightarrow a\|B, B\rightarrow b\|B, B\rightarrow b\}$ is a parallel regular grammar and the corresponding language $L(G)=\{a\|b^{n^{\oplus}}|n\geq 1\}$ is a parallel regular language.\\
\begin{theorem}
$L(G)$ is a parallel regular language if and only if $L$ is recognized by a branching automaton $\mathbf{A}$.
\end{theorem}
\begin{proof}
Let $G=(V,T,P,S)$ be a linear grammar. Then by theorem 1, $L(G)$ is a parallel regular language. We have to prove that $L$ is recognized by a branching automaton $\mathbf{A}$.\\ In otherwords, let $w\in L(G)$ we now to show that $w$ is accepted by a Branching automaton $\mathbf{A}$.\\
Assume that $V=\{V_{0},V_{1},...V_{n}\}$ and $S=V_{0}$. Productions are of the form
$V_{0}\rightarrow a_{1}\|V_{1}$, $V_{1}\rightarrow a_{2}\|V_{2}$,...,$V_{n}\rightarrow a_{l}$, $a_{1},a_{2},...,a_{l}\in \Sigma$.
Let $w\in L(G)$, then the production rules are $V_{0}\Rightarrow a_{1}\|V_{1} \Rightarrow a_{1}\|a_{2}\|V_{2} \Rightarrow ...\Rightarrow a_{1}\|a_{2}\|...\|a_{k}\|V_{n} \Rightarrow a_{1}\|a_{2}\|...\|a_{k}\|a_{l}=w$.\\
A branching automaton [1][6] over the alphabet $\Sigma$ is $\mathbf{A}$=($Q$,$T_{seq}$,$T_{fork}$,\\$T_{join}$,$T_{par}$,$S$,$E$) where $Q$ is the set of finite states. $S$ and $E$ are subsets of $Q$, the set of initial (start) and final (end) states, respectively. $T_{seq}$$\subseteq$$Q$$\times$$\Sigma$$\times$$Q$ is the set of sequential transitions. $T_{par}$$\subseteq$$T_{fork}$$\times$B$\times$$T_{join}$, $B\subseteq \Sigma^{\oplus}$ is the set of parallel transitions where $T_{fork}$$\subseteq$$Q$$\times$$M_{ns}(Q)$ and $T_{join}$$\subseteq$$M_{ns}(Q)$$\times$$Q$ are the set of fork and join transitions. Here $M_{ns}(Q)$ (non empty and non singleton) stands for multisets over $Q$ of cardinality at least 2.\\
The initial state of the automaton is $V_{0}$ and all other $V_{i}$'s are non-terminal states.In otherwords, states of Branching automaton are the variables of the linear grammar.\\ Each production $V_{i}\rightarrow a_{j}\|V_{j}$, for $i=0,1,2,...,(n-1)$ and $j=i+1$ corresponds to the transition in the branching automaton from $V_{i}$ to $V_{j}$ and the transition is defined by $T^{*}_{par}(V_{i}, B)=V_{j}$ , $B\subseteq \Sigma^{\oplus}$, $V_{i}\subseteq T_{fork}$ ,$V_{j}\subseteq T_{join}$ and $T^{*}_{par}(V_{i}, w)=V_{f}$ where $V_{f}$ is a final state. Transitions $T^{*}_{par}$ in branching automaton are given in Figure 1.%for each production $V_{i}\rightarrow a_{1}\|a_{2}\|...\|a_{m}\|V_{j}$ of linear grammar $G$.\\ Similarly, for each production $V_{i}\rightarrow a_{1}\|a_{2}\|...\|a_{m}$ and the corresponding transition of the automaton is $\delta^{*}_{par}(V_{i}, a_{1}\|a_{2}\|...\|a_{m})=V_{f}$ where $V_{f}$ is a final state.
\begin{figure}
        \centering
        \begin{subfigure}[b]{0.475\textwidth}
            \centering
            {\includegraphics[width=5cm, height=2.5cm]{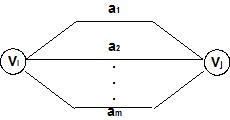}}
            \caption{Represents $V_{i}\rightarrow a_{1}\|a_{2}\|...\|a_{m}\|V_{j}$}
        \end{subfigure}
        \hfill
        \begin{subfigure}[b]{0.475\textwidth}
            \centering
            {\includegraphics[width=5cm, height=2.5cm]{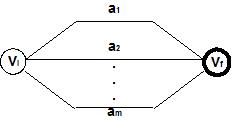}}
            \caption{Represents $V_{i}\rightarrow a_{1}\|a_{2}\|...\|a_{m}$}
        \end{subfigure}
        \caption {\small Transition Representations}
    \end{figure}
Suppose $w\in L(G)$ then it satisfies the above production rules. By the construction of transitions on branching automaton, clearly $V_{f}\in T^{*}_{par}(V_{0},w)$ and hence, $w$ is accepted by $\mathbf{A}$.\\ Conversely, assume $w$ is accepted by branching automaton $\mathbf{A}$. We have to prove that $w$ is in parallel regular language $L(G)$.\\$\mathbf{A}$ recognizes $w$, the automaton passes through a sequence of states $V_{0}, V_{2},...,V_{f}$ using paths labeled by $a_{1}, a_{2},...,a_{m}$.\\ Then $w$ is of the form $w=a_{1}\|a_{2}\|...\|a_{k}\|a_{l}$ and its derivation $V_{0}\Rightarrow a_{1}\|V_{1} \Rightarrow a_{1}\|a_{2}\|V_{2} \Rightarrow ...\Rightarrow a_{1}\|a_{2}\|...\|a_{k}\|V_{n} \Rightarrow a_{1}\|a_{2}\|...\|a_{k}\|a_{l}=w$ exists. Therefore $w\in L(G)$.
\end{proof}
\subsection{Series Parallel Regular Grammar}
A grammar $G$ on series parallel languages is said to be regular if $G$ is right-linear or left-linear or linear.\\
Consider a grammar $G=\{V,T,P,S\}$ with production rules $A\rightarrow xB\|x$ or $A\rightarrow Bx\|x$ or $A\rightarrow x$, where $A,B\in V$ , $x\in SP(T)$.\\ $\emph{Example:}$ Let $G=(V,T,P,S)$ where $V=\{A,S\}$, $T=\{a,b\}$ and \\$P=\{S\rightarrow Aa, A\rightarrow a\|A, A\rightarrow b\}$ is a series parallel regular grammar and $L(G)=\{(a^{n^{\oplus}}\|b)a|n\geq 1\}$ is a series parallel regular language.\\
\begin{theorem}
A language $L\subseteq \Sigma^{\oplus}$ is parallel regular if and only if there exists a parallel regular grammar $G=\{V,T,P,S\}$ such that $L=L(G)$.
\end{theorem}
\begin{proof}
It follows from the theorem 2.
\end{proof}
\begin{theorem}
Let $G=\{V,T,P,S\}$ be a left-linear or right-linear and linear grammar. Then $L(G)$ is a series parallel regular language.
\end{theorem}
\begin{proof}
We prove this theorem by an example.\\ Consider $G=(V,T,P,S)$ where $V=\{A,S\}$, $T=\{a,b\}$ and $P=\{S\rightarrow Aa, A\rightarrow a\|A, A\rightarrow b\}$. This generates the language and $L(G)=\{(a^{n^{\oplus}}\|b)a|n\geq 1\}$ and $L(G)\subseteq SP(\Sigma)$.\\ Here $P$ consists of both left-linear and linear productions. This satisfies the definition of series parallel regular languages. Hence, the theorem proved.
\end{proof}
\begin{theorem}
A language $L\subseteq SP(\Sigma)$ is regular if and only if there exists a series parallel regular grammar $G=\{V,T,P,S\}$ such that $L=L(G)$.
\end{theorem}
\begin{proof}
It is immediate from the above theorem 4.
\end{proof}

\end{document}